\documentclass[letterpaper, 10 pt]{ieeeconf} 
\usepackage{amsmath,amssymb,amsfonts}
\usepackage{graphicx}
\usepackage[dvipsnames]{xcolor}
\usepackage{enumerate}
\usepackage{wrapfig,lipsum}
\usepackage{bbm}
\usepackage{cite}
\usepackage{tikz}
\usepackage{algorithm}
\usepackage{algpseudocode}
\usepackage{tikz-network}
\usepackage{balance}
\usetikzlibrary{positioning}
\usepackage[bb=boondox]{mathalfa}

\newcommand{\Ec}{\mathcal{E}}

\newcommand{\Gc}{\mathcal{G}}

\newcommand{\Rc}{\mathcal{R}}
\newcommand{\Mc}{\mathcal{M}}
\newcommand{\Sc}{\mathcal{S}}

\newcommand{\Lc}{\mathcal{L}}
\newcommand{\Tc}{\mathcal{T}}
\newcommand{\Wc}{\mathcal{W}}
\newcommand{\Uc}{\mathcal{U}}
\newcommand{\Vc}{\mathcal{V}}
\newcommand{\Xc}{\mathcal{X}}

\newcommand{\Rb}{\mathbb{R}}

\newcommand{\diag}{{\rm diag}}

\usetikzlibrary{decorations.pathmorphing}
\usetikzlibrary{matrix,fit}
\pgfkeys{tikz/mymatrix/.style={matrix of math nodes,inner sep=0pt,row sep=0em,column sep=0em,nodes={inner sep=6pt}}}

\newcommand{\exampleend}{\hfill $\triangle$}

\usetikzlibrary{arrows}
\newtheorem{lemma}{Lemma}
\newtheorem{remark}{Remark}
\newtheorem{theorem}{Theorem}
\newtheorem{definition}{Definition}
\newtheorem{example}{Example}

\newtheorem{proposition}{Proposition}
\newtheorem{corollary}{Corollary}

\title{\textbf{Finding Conditions for Target Controllability under Christmas Trees}}
\author{Marco Peruzzo, Giacomo Baggio, Francesco Ticozzi%
\thanks{The authors are with the Department of Information Engineering, University of Padova, Italy. Emails: \texttt{marco.peruzzo.5@phd.unipd.it}, \texttt{baggio@dei.unipd.it}, \texttt{ticozzi@dei.unipd.it}.}
}

\begin{document}
\maketitle
\thispagestyle{empty}
\pagestyle{empty}

 \begin{abstract} 
 This paper presents new graph-theoretic conditions for structural target controllability of directed networks. After reviewing existing conditions and highlighting some gaps in the literature, we introduce a new class of network systems, named Christmas trees, which generalizes trees and cacti. We then establish a graph-theoretic characterization of sets of nodes that are structurally target controllable for a simple subclass of Christmas trees. Our characterization applies to general network systems by considering spanning subgraphs of Christmas tree class and allows us to uncover target controllable sets that existing criteria fail to identify.
 \end{abstract}
 
\section{Introduction}

Network systems are widespread in nature, with complex phenomena often emerging from the interactions of numerous simple dynamical units. In recent years, significant efforts have been made to describe and engineer these systems, leading to advances in understanding their structure and dynamics \cite{SS:01,YYL-JJS-ALB:11,FB:24}. The network formalism has proven successful in modeling phenomena across different fields, including social sciences,  biology,~and~engineering~\cite{PA-TR:17,bernardo2021achieving,TV-KR-LA:06,dorfler2013synchronization}.

One of the key challenges for the control community is the controllability analysis of large-scale network systems. Effectively controlling these systems is difficult, as the required control energy can scale exponentially with the size of the network \cite{SZ-FP-FB:2014,GB-SZ:24}. Moreover, the knowledge of system parameters is often not precise and the structure of the system significantly impacts the possibility to steer the system to a target state. To address the latter problem, recent works \cite{YYL-JJS-ALB:11,pequito2015framework,mousavi2017structural} have resorted to the framework of structural controllability introduced by Lin in the 1970s \cite{CTL:74}.
Unlike classical controllability, structural controllability is determined only by the sparsity patterns of system matrices, without requiring knowledge of their exact numerical values.

However, in many applications, it is not necessary to control the entire state of the network. Instead, it is often sufficient to control only a subset of critical nodes. From a control-theoretic viewpoint, this problem is addressed through the notions of output controllability and its special case, target controllability \cite{KM:1990,AB:2014}. 
Although not a new problem \cite{KM:1990}, finding conditions for classical and structural output controllability remains an active area of research, with several studies appearing in recent years \cite{moothedath2019target,BD-JL-MJ:23,AM-CD-MA:2023}. To date, and to the best of our knowledge, only partial graph-theoretic characterizations of structural output and target controllability have been established \cite{SH:1980,AB:2014,moothedath2019target,AM-CD-MA:2023}, with a complete characterization available only for the special case of systems described by~undirected~networks~\cite{LJ-CX-PS:21}.

\noindent\textbf{Contribution.} The contributions of the paper are threefold.
First, we provide an overview of existing graph-theoretic conditions for structural target controllability, highlighting some misleading conjectures, for which we present simple counterexamples. Second, we introduce a new class of structured network systems called Christmas trees. These systems generalize trees and cacti and represent one of the simplest classes of network systems that can span any input-accessible network, including non-structurally controllable ones.
We then focus on a specific subclass of Christmas trees, termed elementary Christmas trees, and derive graph-theoretic conditions for their structural target controllability. Third and finally, we discuss how these new conditions generalize existing ones and illustrate through examples how they enable the identification of sets of target controllable nodes that cannot be detected using~previously~established~criteria.

\noindent\textbf{Organization.}
The rest of the paper is organized as follows. In Sec.~\ref{sec:preliminaries}, we describe the systems considered in this paper and provide some background on their graph-theoretic description and controllability. In Sec.~\ref{sec:existing}, we discuss structural target controllability and review existing graph-theoretic conditions. In Sec.~\ref{sec:christmas}, we introduce Christmas trees and elementary Christmas trees. In Sec.~\ref{sec:target-christmas}, we derive conditions for structural target controllability of elementary Christmas trees and illustrate through examples how these conditions differ from existing ones. Finally, Sec.~\ref{sec:conclusions} concludes the paper with some remarks and directions for future research.

\noindent\textbf{Notation.}
We denote with $\mathbb{R}^{n\times m}$ the set of $n\times m$ matrices with real entries. $[A]_{ij}$, $[A]_{s,:}$ are respectively, the $(i,j)$-th entry and the $s$-th row of of matrix $A\in \mathbb{R}^{n\times m}$. We let $\mathrm{diag}\{a_1,\dots,a_n\}$ indicate the diagonal matrix with diagonal entries $a_1$, \dots, $a_n\in\mathbb{R}$. We use calligraphic letters to denote sets and, for a set \(\mathcal{Z}\), the symbol \(|\mathcal{Z}|\) indicates its cardinality. Finally, for $a,r \in \mathbb{R}$, $\lceil a \rceil$ and $a \mod r$ denote respectively the ceil function, mapping $a$ to the least integer greater than or equal to $a$, and the remainder of the division of $a$ by $r$.

\thispagestyle{empty}
\pagestyle{empty}

\section{Preliminaries}\label{sec:preliminaries}
We consider a linear time-invariant system
\begin{align}\label{eq:sys}
\Sigma: \begin{cases}\dot x(t) = A x(t) + Bu(t) &\\
y(t) = Cx(t) &
\end{cases}
\end{align}
where  $x(t)\in\Rb^n$, $u(t)\in\Rb^m$, and $y(t)\in\Rb^p$ are the state, input, and output of the system at time $t$, respectively, and matrices $A\in\Rb^{n\times n}$, $B\in\Rb^{n\times m}$, and $C\in\Rb^{p\times n}$ are the state, input, and output matrices of the system, respectively.

Equation \eqref{eq:sys} can be interpreted as a network system composed by first-order linear time-invariant subsystems: each component of $x(t)$ represents the state of one of these subsystems.

The system \eqref{eq:sys} is output controllable if for all vectors $y_f\in\Rb^p$ there exists a control input $u(t)$ such that $y(t^*)=y_f$ for some $t^*>0$. 
A well-known condition to assess output controllability relies on the $k$-steps output controllability matrix
\begin{equation}\label{eq:Rc}
\Rc_k = C[B\ AB\ A^2B\ \cdots \ A^{k-1}B ].
\end{equation}
Specifically, \eqref{eq:sys} is output controllable if and only if $\Rc_k$ is full row-rank for some $k\ge n$. While verifying this condition for $k=n$ is sufficient, in some parts of the paper it will be convenient to consider a general $k\ge n$.
A special case of output controllability is target controllability, where the rows of $C$ are constrained to be canonical vectors in $\Rb^n$. This choice of output corresponds to selecting only the
states of a subset of nodes as targets for the control.

In this work, we treat \eqref{eq:sys} as a structured system, that is, a system in which matrices $A$, $B$, $C$ are structured: their entries are either zero or generic real numbers.
To this structured system, we can associate a graph $\Gc=(\Vc,\Ec)$, with $\Vc=\Uc\cup\Xc$, $\Uc =\{u_1,\dots,u_m\}$\footnote{If $\Uc$ consists of a single node, we denote the input node simply by $u$.}, $\Xc=\{1,\dots,n\}$ being the set of input and state nodes, respectively, and $\Ec=\Ec_x\cup \Ec_u$, with $\Ec_x=\{(i,j)\,:\, A_{ji}\ne 0\}$, $\Ec_u=\{(u_i,j)\,:\, B_{ji}\ne 0\}$.

In the remainder of the paper we will often make use of the following definitions which apply to the graph $\Gc$ of a structured system. A \textit{path} $p=\{w_1,\dots,w_s\}\subseteq \Vc$ is a sequence of nodes such that $(w_j,w_{j+1})\in{\Ec}$ for all $ j\in\{1,\dots,s-1\}$. Moreover, the \textit{length} of $p$ is $\ell=s-1$, where $s$ is the number of nodes in the path. Finally, the \textit{weight} of $p$ is the product of the weights of all the edges $(w_j,w_{j+1})\in{\Ec}$ for all $ j\in\{1,\dots,s-1\}$. We distinguish the following classes of paths. A {\em stem} is a path whose nodes are all distinct and such that $w_1\!=\!u_i \in {\cal U}$, $w_j=v\in{\cal X}$. A \textit{cycle} is a path with all distinct nodes, except $w_1\!=\!w_s.$
Among all the paths connecting two nodes, a \textit{shortest path} is a path with minimum number of nodes.

 To study target controllability, we require input accessibility. The structured system \eqref{eq:sys} is \textit{input accessible} if, for every state node $i$,  $i = 1, \dots, n$, there exists a stem ending in $i$. For single-input, input-accessible systems, we label with $\ell_{\text{min}}(x)$ the length of a shortest path from the input node $u$ to $x\in\Xc$ and with $L(\Sc)$ the set of lengths of the paths reaching nodes in $\Sc\subseteq\Xc$ from $u$.

Finally, several conditions in structural controllability theory are given in terms of spanning subgraphs. A set of nodes $\Sc$ is \textit{covered} (or spanned) by a subgraph $\Gc'=(\Vc',\Ec')$ of ${\Gc}$ if ${\cal S}\subseteq$ $\Vc'$, and a graph is covered by a subgraph if the subgraph and the graph have the same vertex set.

\section{Structural target controllability and existing graph-theoretic results}\label{sec:existing}

In this section, we review the notions of structural output and target controllability and discuss existing graph-theoretic conditions to verify such properties.

\begin{definition}[Structural output controllability]
Consider the structured system \eqref{eq:sys}. The system is said to be structurally output controllable if there exists a numerical realization of $A$, $B$, $C$ such that the output controllability matrix $\Rc_k$ in \eqref{eq:Rc} has full row-rank for some $k\geq n$.
\end{definition}

The term structural reflects the fact that structural output controllability is a generic property of \eqref{eq:sys}: if a system is structurally output controllable, then it is output controllable for almost any choice of the generic entries of $A$, $B$, $C$.  
A special case of structural output controllability is structural target controllability.

\begin{definition}[Structural target controllability]
Consider the system \eqref{eq:sys}. The system is said to be structurally target controllable if (i) the rows of $C$ are canonical vectors in $\Rb^n$ (ii) there exists a numerical realization of $A$, $B$ such that ${\Rc}_k$ in \eqref{eq:Rc} has full-row rank for some $k\geq n$.
\end{definition}

With a slight abuse of terminology, in the following we will say that a set $\Tc\subseteq \Xc$ is structurally target controllable, if \eqref{eq:sys} is structurally target controllable for an output matrix $C$ whose rows have non-zero entries indexed by $\Tc$.

Although the concept of structural output controllability is well established, characterizing it in graph-theoretic terms has proven challenging. In particular, \cite{KM:1990} provides a partial characterization by establishing lower and upper bounds on the generic dimension of the output controllable subspace, and presents a counterexample to a conjectured necessary and sufficient condition.

In contrast, more results are available for structural target controllability, some of which emerged only recently. The following is a classic result which follows from the graph-theoretic characterization of the generic dimension of the controllability subspace of the system in \cite{SH:1980}.

\begin{proposition}[Stem-cycle condition]\label{prop:stemcycle}
Consider the structured system \eqref{eq:sys} and assume it is input accessible. A set of nodes $\Tc\subseteq \Xc$ is structurally target controllable if it can be covered by a disjoint union of stems and cycles in $\mathcal{G}$.
\end{proposition}

A similar result can be derived considering a disjoint union of covering cacti by dropping the assumption on input accessibility. 
Cacti are a class of graphs associated to structured systems that were introduced in the 1970s, and have proven fundamental to develop many graph-theoretic conditions in structural system theory, see, e.g., \cite{CTL:74} for a detailed description of these graphs.

More recently, the following alternative graph-theoretic condition for a set of nodes to be structurally target controllable has been established in \cite{AB:2014}. 

\begin{proposition}[$k$-walk condition]\label{prop:kwalk}
Consider the system \eqref{eq:sys}. Assume that the system is input accessible and has a single input ($m=1$). A set $\Tc\subseteq \Xc$ is structurally target controllable if $\ell_{\text{min}}(t)\ne\ell_{\text{min}}(s)$ for all $t,s\in\Xc$, $t\neq s$.
\end{proposition}

Although applicable to any network system, this result provides a complete characterization of the sets of target controllable nodes only for directed tree networks \cite{AB:2014}.

In \cite{moothedath2019target}, a necessary graphical condition for a set of nodes to be structurally target controllable is presented. This condition is based on the existence of a perfect matching\footnote{
Given a bipartite graph $\Gc=(\Vc\cup \Wc,\Ec)$, $\Vc\cap \Wc=\emptyset$ and $\Ec=\Vc\times \Wc$, a matching is a collection of edges $\Mc\subseteq \Ec$ such that no two edges in $\Mc$ share a common endpoint. For $|\Vc|\leq |\Wc|$, $|\Vc|\geq |\Wc|$, respectively,  the matching is said to be perfect if $|\Mc|=|\Wc|, \ |\Mc|=|\Vc|$.}
in a bipartite graph $\Gc_B=(\Tc \cup \Lc, \Ec_B)$, where $\Tc$ contains the target nodes and $\Lc$ the lengths of paths from input nodes to target nodes, up to a maximum length equal to the generic dimension of the controllability subspace. For any $t\in\Tc$ and $\ell\in\Lc$, we have an edge $(t,\ell)\in\Ec_B$ if there exists a path of length $\ell$ from an input node to $t$.
The condition is conjectured to be sufficient as well, and an algorithm is proposed to verify target controllability based on this assumption. However, as demonstrated in the following example, this condition~is~not~sufficient.

\begin{example}[Counterexample to conjecture in \cite{moothedath2019target}]\label{ex:counter_example}
  We consider the system $\Sigma$ described by the following state, input, and output matrices:
  \begin{align*}
    A=\left[\begin{smallmatrix}
    \setlength{\tabcolsep}{0pt}
0 & a_{12} & 0 & 0 & 0\\
a_{21}  & 0 & 0 & 0 & 0\\
0 & 0 & 0 & 0 & a_{35} \\
0 & 0 & a_{43}  & 0 & 0\\
0 & 0 & a_{53}  & 0 & 0
\end{smallmatrix}\right], \ 
B=\left[\begin{smallmatrix}
    b_1 \\ 0 \\ b_3 \\ 0 \\ 0 
\end{smallmatrix}\right], \ 
C=\left[\begin{smallmatrix}
    1 &  0 & 0  &0 &0\\ 
    0 &  0 & 1 & 0 & 0\\
    0 &  0 & 0 &1 &0\\ 
    0 &  0 & 0 & 0 & 1 
\end{smallmatrix}\right]
\end{align*}
The target set is $\Tc=\{1,3,4,5\}.$ Fig.~\ref{fig:counter_example}(a) shows the graph $\mathcal{G}$ associated with $\Sigma$ while Fig.~\ref{fig:counter_example}(b) the bipartite graph introduced in \cite{moothedath2019target}. The latter graph presents a perfect matching. Thus, according to the algorithm proposed in \cite{moothedath2019target}, $\Tc$ is structurally target controllable. 
However, by computing the output controllability matrix of $\Sigma$ for $k=n=5$,
\begin{align*}
     \Rc_n =\Bigg[\begin{smallmatrix}
     b_1 & 0 &  a_{12} \, a_{21} \, b_1& 0 &  a_{12}^2 \, a_{21}^2 \,b_1\\
     b_3 & 0 & a_{35} \,a_{53} \,b_3  & 0 & \,{a_{35} }^2 \,{a_{53} }^2 \,b_3 \\
     0 &  a_{43} \,b_3  & 0 & a_{43} \, a_{35} \,a_{53} \,b_3  & 0\\
     0 & a_{53} \,b_3  & 0 & a_{35} \,{a_{53} }^2 \,b_3  & 0 
   \end{smallmatrix}\Bigg],
 \end{align*}
it is possible to recognize that $[\Rc_n]_{3,:}= a_{43} [\Rc_n]_{4,:}/({ a_{53}})$.
Therefore, $\Rc_n$ is not full row-rank for any choice of the generic parameters of the system. This implies that the set $\Tc$ is not structurally target controllable.\exampleend
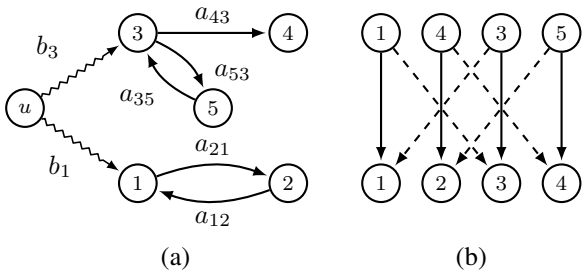
\begin{figure}
    \begin{tikzpicture}[shorten >=1pt, auto, ultra thick,
   node_style/.style={draw, circle,thick, fill=white, minimum size=0.5cm,font=\footnotesize,inner sep=0cm},every edge/.append style = {thick}]
        
        \node[node_style] (u) at (-0.5,0) {$u$};
        \node[node_style] (1) at (1,-1) {$1$};
        \node[node_style] (2) at (3,-1) {$2$};
        \node[node_style] (3) at (1,1) {$3$};
        \node[node_style] (4) at (3,1) {$4$};
        \node[node_style] (5) at (2,0) {$5$};

        \draw[-latex,semithick,line join=round,
             decorate, decoration={
                zigzag,
                segment length=4,
                amplitude=.9,
                post=lineto,
                post length=7.5pt
            }] (u) -- (1) node [midway,label=below left:{$b_1$}] {};

            \draw[-latex,semithick,line join=round,
             decorate, decoration={
                zigzag,
                segment length=4,
                amplitude=.9,
                post=lineto,
                post length=7.5pt
            }] (u) -- (3) node [midway,pos=0.2,label=:{$b_3$}] {};
            
        \draw[-latex] (1) edge[above,bend left, out=20, in=160] node{$a_{21}$} (2);
        \draw[-latex] (2) edge[below, bend left, in=160, out=20] node{$a_{12}$} (1);
        \draw[-latex] (3) edge node{$a_{43}$} (4);
        
        \draw[-latex] (3) edge[bend left, above,at end,anchor=south west,  out=20, in=160 ] node{$a_{53}$} (5);
        \draw[-latex] (5) edge[bend left, out=20, in=160] node{$a_{35}$} (3);
        \node at (1.5,-2) 
    {(a)};
    \begin{scope}[xshift=120px]
    
        \node[node_style] (n1) at (0,1) {$1$};
        \node[node_style] (n4) at (0.8,1) {$4$};
        \node[node_style] (n3) at (1.6,1) {$3$};
        \node[node_style] (n5) at (2.4,1) {$5$};
        
        \node[node_style] (1) at (0,-1) {$1$};
        \node[node_style] (2) at (0.8,-1) {$2$};
        \node[node_style] (3) at (1.6,-1) {$3$};
        \node[node_style] (4) at (2.4,-1) {$4$};

        \draw[-latex] (n4) edge node{} (2);
        \draw[-latex,dashed] (n4) edge node{} (4);
        \draw[-latex,dashed] (n5) edge node{} (2);
        \draw[-latex] (n5) edge node{} (4);
        \draw[-latex] (n1) edge node{} (1);
        \draw[-latex, dashed] (n3) edge node{} (1);
        \draw[-latex, dashed] (n1) edge node{} (3);
        \draw[-latex] (n3) edge node{} (3);
        
        \node at (1.2,-2) {(b)};
    \end{scope}
\end{tikzpicture}
    \vspace{-0.5cm}
    \caption{Counterexample to the sufficiency of condition in \cite{moothedath2019target}. Fig.~\ref{fig:counter_example}(a) shows the graph describing a structured system which is not structurally target controllable with target set $\Tc=\{1,3,4,5\}$. The generic dimension of the controllability subspace is $4$. In Fig.~\ref{fig:counter_example}(b) it is depicted the bipartite graph introduced in the aforementioned work to check structural target controllability of $\Sigma$. The upper nodes corresponds to $\Tc$, the bottom nodes are associated to path lengths $\{1,2,3,4\}$. The graph presents a perfect matching. The edges of the perfect matching are solid, the edges not considered in the matching are dashed.}
    \label{fig:counter_example}
\end{figure}
\end{example}

A necessary and sufficient graph-theoretic condition for structural target controllability has appeared in \cite{AM-CD-MA:2023}. However, this condition applies only to a specific class of structured systems, where either (i) there exists a numerical realization of $A$ (with non-zero generic weights) that is diagonalizable, or (ii) each diagonal entry of $A$ is generic. These assumptions limit the applicability of the result: for instance, it cannot be applied to systems described by directed acyclic graphs, since the adjacency matrix of such graphs is a nilpotent matrix \cite[Sec.~6.4.1]{newman2018networks}. 

Finally, another necessary and sufficient graph-theoretic condition for structural target controllability has been established in \cite{LJ-CX-PS:21}. This condition applies to symmetric state matrices $A$, an assumption we do not make in this paper.

\section{Christmas trees}\label{sec:christmas}
To establish alternative, more general graph-theoretic conditions for structural target controllability, we introduce a new class of structured systems. These systems are described by graphs that are neither cacti nor directed trees, but instead exhibit a ``hybrid'' topology between the two.
\begin{definition}[Christmas tree]\label{def:xmas}
    A Christmas tree is a single-input, input-accessible system \eqref{eq:sys} whose graph $\Gc$ satisfies the following conditions:
    \begin{enumerate}[C1)]
        \item each state node of $\Gc$ belongs to at most one cycle;
        \item for each cycle of $\Gc$, exactly one node has in-degree 2;
        \item all other state nodes of $\Gc$, whether in cycles or not, have in-degree 1. 
    \end{enumerate}\smallskip
\end{definition}

In intuitive terms, a Christmas tree is described by a graph $\Gc$ that combines cycles (Christmas baubles) and tree-like branches. 
Within each cycle, only one node has two incoming edges. This node serves as a ``junction'' that connects the cycle to the rest of the graph. Observe that if the graph $\Gc$ describing a Christmas tree contains no cycles then it coincides with a directed tree. Additionally, if $\Gc$ can be covered by a disjoint union of stems and cycles then it corresponds to a cactus. An example of graph describing a Christmas tree network is depicted in Fig.~\ref{fig:ct}.

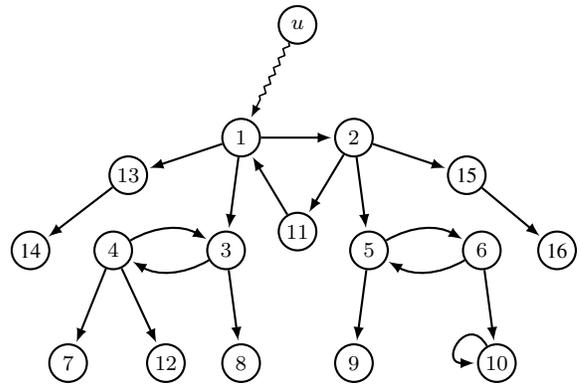
\begin{figure}
    \centering
        \begin{tikzpicture}[shorten >=1pt, auto, ultra thick,
   node_style/.style={draw, circle,thick, fill=white, minimum size=0.5cm,font=\footnotesize,inner sep=0cm},every edge/.append style = {thick}]
   
  \node[node_style] (u) at (1.75,1.5) {$u$};
  \node[node_style] (1) at (1,0) {$1$};
  \node[node_style] (2) at (2.5,0) {$2$};
  \node[node_style] (3) at (1-0.2,-1.5) {$3$};
  \node[node_style] (4) at (-0.5-0.2,-1.5) {$4$};
  \node[node_style] (11) at (1.75,-1.25) {$11$};

\node[node_style] (5) at (2.5+0.2,-1.5) {$5$};
  \node[node_style] (6) at (4+0.2,-1.5) {$6$};

\node[node_style] (7) at (-0.5-0.8,-3) {$7$};
\node[node_style] (8) at (1,-3) {$8$};
\node[node_style] (9) at (2.5,-3) {$9$};
\node[node_style] (10) at (4+0.4,-3) {$10$};
\node[node_style] (12) at (-0,-3) {$12$};

\node[node_style] (13) at (-0.5,-0.5) {$13$};

\node[node_style] (15) at (4,-0.5) {15};
\node[node_style] (14) at (-1.8,-1.5) {14};
\node[node_style] (16) at (5.2,-1.5) {16};

  \draw[-latex] (1) edge node{} (2);
  \draw[-latex] (2) edge node{} (11);
  \draw[-latex] (11) edge node{} (1);
  
  \draw[-latex] (3) edge [bend left] node{} (4);
  \draw[-latex] (4) edge [bend left]  node{} (3);
  \draw[-latex] (4) edge node{} (12);
  
   \draw[-latex] (5) edge [bend left] node{} (6);
  \draw[-latex] (6) edge [bend left]  node{} (5);
  
  \draw[-latex] (1) edge node{} (13);
  \draw[-latex] (13) edge node{} (14);

  \draw[-latex] (1) edge node{} (3);
  \draw[-latex] (2) edge node{} (5);
  \draw[-latex] (3) edge node{} (8);
  \draw[-latex] (4) edge node{} (7);
  \draw[-latex] (5) edge node{} (9);
  \draw[-latex] (6) edge node{} (10);
  \draw[-latex] (10) edge [loop,looseness=5,in=180, out=120] node{} (10);
  \draw[-latex] (2) edge node{} (15);
  \draw[-latex] (15) edge node{} (16);

   \draw[-latex,semithick,line join=round,
             decorate, decoration={
                zigzag,
                segment length=4,
                amplitude=.9,
                post=lineto,
                post length=7.5pt
            }] (u) -- (1);
    
\end{tikzpicture}
    \vspace{-0.5cm}
    \caption{Example of graph $\Gc=(\Wc,\Ec)$ of a Christmas tree. }
    \label{fig:ct}
\end{figure}
 
 \begin{definition}
 Consider a directed graph $\Gc = (\Vc,\Ec)$. We say that a node $v\in \mathcal{V}$ is backward-connected to a cycle $c$ if either $v\in c$ or there exists a path $p=\{w_1,\dots,w_s\}$ in $\Gc$ such that $w_1\in c$ and $w_s= v$.
 \end{definition}

In the remainder of the paper, we focus on a specific subclass of Christmas trees that satisfies two additional properties.

\begin{definition}[Elementary Christmas tree]
    An elementary Christmas tree is a Christmas tree whose graph $\Gc$ satisfies the following conditions:
    \begin{enumerate}[C4)]
    \item[C4)] all the cycles of $\Gc$ have the same length $r\geq 1$;
    \item[C5)] each state node of $\Gc$ is backward-connected to at most one cycle.
\end{enumerate}\smallskip
\end{definition}

Given an elementary Christmas tree, we define its $i$-th periodic class as\footnote{If no cycles are present, we define a single class $\Xc_0 = \Xc$.}
$$
\Xc_i = \{x\in\Xc \,:\, \ell_{\text{min}}(x) \ \mathrm{mod} \ r=i \}, \ \ i=0,\dots,r-1.
$$
In view of conditions C4)-C5), the following properties hold for an elementary Christmas tree:
\begin{enumerate}[P1)]
    \item $\Xc=\Xc_0\cup\cdots \cup \Xc_{r-1}$, ${\Xc}_i\cap {\Xc}_j=\emptyset$ and $L(\Xc_i)\cap L(\Xc_j)=\emptyset$ for all $i,j$, $i\neq j$; \label{prop:P1}
    \item for each node $x\in \Xc$ and each $\ell\in L(x)$, there is a unique path of length $\ell$ connecting the input node with node $x$.\label{prop:P2}
\end{enumerate}

We define the shortest distance in $\Xc_i$ from the input to $x\in \Xc_i$ as\footnote{If no cycles are present, we define $\ell_{\text{min}}^{(i)}(x)= \ell_{\text{min}}(x)$.}
$$
\ell_{\text{min}}^{(i)}(x) = \lceil\ell_{\text{min}}(x)/r\rceil.
$$
Note that $\ell_{\text{min}}^{(i)}(x)$ equals the length of a shortest path from $u$ to $x$ when only nodes in $\Xc_i$ are counted.

Finally, if a node $w \in \Xc_i$ is backward-connected to a cycle, with a slight abuse of nomenclature, we say that $v \in \Xc_i$ is a child of $w$ if:
\begin{enumerate}
    \item $\ell_{\text{min}}(v) = \ell_{\text{min}}(w) + r$,
    \item $v$ is backward-connected to the same cycle of $w$.
\end{enumerate}
If $w$ is part of a path, 
a child of $w$ is the first node in the path after $w$ that belongs to the same periodic class as $w$. However, notice that, in general, not all children are forced to belong to a path containing $w$.

\section{Target controllability of elementary Christmas trees}\label{sec:target-christmas}

In this section, we establish a graphical characterization of the set of structurally target controllable nodes in elementary Christmas trees. Our main result relies on the following instrumental lemma, which sheds light on the structure of the output controllability matrix of an~elementary~Christmas~tree.

\begin{lemma}{\em (Output controllability matrix of elementary Christmas trees)}\label{lem:ctrb_matrix_struct}
Consider an elementary Christmas tree and let $\Rc_k$ denote its output controllability matrix with output matrix $C$ identified by the target node set $\Tc=\Tc_0\cup\cdots\cup \Tc_{r-1}$, $\Tc_i=\{t_{1}^{(i)},\dots,t_{s_i}^{(i)}\}\subseteq \Xc_i .$ Then, $\forall k\geq n$,
there exist permutation matrices $P_1$, $P_2$ such that
\begin{align}\label{eq:hatRo}
        \hat{\Rc}_k=P_1\Rc_k P_2=\begin{bmatrix}
        D^{(0)}\Rc_{k}^{(0)}  & \dots & 0 & 0 \\
        \vdots & \ddots & \vdots & \vdots
        \\
            0 & \dots & D^{(r-1)}\Rc_{k}^{(r-1)} & 0
        \end{bmatrix},
\end{align}
where $D^{(i)}\in\Rb^{s_i\times s_i}$ are diagonal matrices of the form
\begin{align}\label{eq:Di}
D^{(i)}=\textrm{diag}\{\beta_1^{(i)},\dots,\beta_{s_i}^{(i)}\},
\end{align}
with $\beta_j^{(i)}$ being the weight of the shortest path in $\Gc$ connecting the input node to the state node $t^{(i)}_j$, 
and the rows of the block $\Rc_{k}^{(i)}$ are of the following type: 
\begin{enumerate}[T1)]
    \item  if $t^{(i)}_j$ is not backward-connected to cycles,
    $$[\Rc_{k}^{(i)}]_{j,:}=[0\ \cdots \ 0\  1  \ 0 \ \cdots \ 0];$$
    \item if $t^{(i)}_j$ is backward-connected to a cycle with weight $\alpha$,
    \begin{align*}
        [\Rc_{k}^{(i)}]_{j,:}=[0\ \cdots \ 0\  1  \ \alpha \ \alpha^2 \ \alpha^3\ \cdots].
    \end{align*}
\end{enumerate}
In both cases T1, T2 the first non-zero entry appears in position $\ell_{\text{min}}^{(i)}(t^{(i)}_j)$.

\begin{proof}
First, since all sets $\Tc_i$ are disjoint, we can reorder the rows of $\Rc_k$ using a permutation matrix $P_1$ to group the rows corresponding to the nodes in $\Tc_0$ first, followed by those corresponding to $\Tc_1$, and so on, up to $\Tc_{r-1}$. 

Next, note that the entries of $\Rc_k$ admits the following graph-theoretic interpretation: in the row of $\Rc_k$ corresponding to node $t \in \Tc$, the $i$-th entry is the sum of the weights of all paths of length $i$ connecting the input node $u$ to $t$.\footnote{This follows from the fact that $[A^s]_{ij}$ equals the sum of the weights of paths of length $k$ from node $j$ to node $i$. } This observation, combined with property P\ref{prop:P1}, implies that each column of $\Rc_k$ can have non-zero entries only in rows corresponding to nodes from at most one set $\Tc_i$, where the latter set may vary from column to column. Therefore, we can reorder the columns of $P_1 \Rc_k$ using a permutation matrix $P_2$ to obtain a matrix with the block structure as in \eqref{eq:hatRo}. 
Note that the zero columns of $\hat{\Rc}_k$ are associated with lengths of paths from $u$ to nodes in $\Xc \setminus \Tc$.

Let $\hat{\Rc}_{k}^{(i)}$ denote the $i$-th diagonal block of $\hat{\Rc}_k$, whose rows correspond to nodes of $\Tc_i$. By property P\ref{prop:P2} and the aforementioned graphical interpretation of the entries of $\Rc_k$,
\begin{enumerate}[i)]
\item if $t^{(i)}_j\in\Tc_i$ is not backward-connected to cycles there is a single path connecting $u$ to $t^{(i)}_j$. This path has length $\ell_{\text{min}}(t^{(i)}_j)\leq n$ and weight $\beta_j^{(i)}$. Since $k\geq n$, the row of $\hat{\Rc}_{k}^{(i)}$ corresponding to this node is different from the zero vector: it has $\beta_j^{(i)}$ in position $\ell_{\text{min}}^{(i)}(t^{(i)}_j)$ and all the~other~entries~$0$.
\item if $t^{(i)}_j \in \Tc_i$ is backward-connected to a cycle with weight $\alpha$, there are multiple paths connecting $u$ to $t^{(i)}_j$. Specifically, there is a unique path of length $\ell_{\text{min}}(t^{(i)}_j) + s r$ and weight $\beta_j^{(i)} \alpha^s$ for all $s = 0, 1, 2, \dots$. The row of $\hat{\Rc}_{k}^{(i)}$ corresponding to this node has $\beta_j^{(i)}\alpha^s$ in position $\ell_{\text{min}}^{(i)}(t^{(i)}_j)+s$ and all the other entries $0$. Since $k\geq n$ and $\ell_{\text{min}}(t^{(i)}_j)\leq n$, the row of $\hat{\Rc}_{k}^{(i)}$ corresponding to $t^{(i)}_j$ is always different from the zero vector.
\end{enumerate}

Based on the above observations, it follows that $\hat{\Rc}_{k}^{(i)}=D^{(i)}\Rc_{k}^{(i)}$, where $D^{(i)}$ is defined in \eqref{eq:Di} and the rows of $\Rc_{k}^{(i)}$ have the form T1 or T2. 
\end{proof}
\end{lemma}

\begin{example}
To illustrate the decomposition of $\Rc_k$ in Lemma \ref{lem:ctrb_matrix_struct}, we consider the elementary Christmas tree of Example \ref{ex:counter_example}. 
In this case, $r=2$ and the periodic~classes~are
$$
\Xc_0 = \{1,3\}, \ \ \Xc_1 = \{2,4,5\}. 
$$
The set of target nodes is
$$
\Tc= \Tc_0\cup\Tc_1, \ \ \Tc_0=\{4,5\},\ \Tc_1=\{1,3\}.
$$
The weights of the shortest paths from $u$ to node $4,5,1,3$ are $\beta^{(0)}_1= b_3a_{43}$, $\beta^{(0)}_2= b_3 a_{53}$, $\beta^{(1)}_1=b_1$, $\beta^{(1)}_2=b_3$, respectively. In addition, we let $\alpha_1=a_{12} a_{21}$, $\alpha_2=a_{35} a_{53}$ denote the weights of the cycles $\{1,2,1\}$, $\{3,5,3\}$, respectively.
    
It holds
\begin{align*}
        \hat{\Rc}_n\!=\!\underbrace{\left[\begin{smallmatrix}
            0 & 0 & 1 & 0 \\
            0 & 0 & 0 & 1 \\
            1 & 0 & 0 & 0 \\
            0 & 1 & 0 & 0 
        \end{smallmatrix}\right]}_{P_1}
            \Rc_n
            \underbrace{\left[\begin{smallmatrix}
            0 & 0 & 1 & 0 & 0 \\
            1 & 0 & 0 & 0 & 0 \\
            0 & 0 & 0 & 1 & 0 \\
            0 & 1 & 0 & 0 & 0 \\
            0 & 0 & 0 & 0 & 1
        \end{smallmatrix}\right]}_{P_2} =  \left[\begin{smallmatrix}
                D^{(0)} {\Rc}_n^{(0)}& 0\\
            0 & D^{(1)} \Rc_n^{(1)}
            \end{smallmatrix}\right]
    \end{align*}
where $D^{(0)}=\diag\{\beta^{(0)}_1,\beta^{(0)}_2\}$, $D^{(1)}=\diag\{\beta^{(1)}_1,\beta^{(1)}_2\}$, 
\begin{align*}
\Rc_n^{(0)}= \left[\begin{smallmatrix}
                1 & \alpha_2\\
                1 & \alpha_2\\
            \end{smallmatrix}\right],\ 
\Rc_n^{(1)} = \left[\begin{smallmatrix}
                1 & \alpha_1 &\alpha_1^2 \\
                1 & \alpha_2 & \alpha_2^2\\
            \end{smallmatrix}\right].
\end{align*}
\exampleend 
\end{example}

Given a set of nodes $\Tc_i \subseteq \Xc_i$, we next introduce a partition of $\Tc_i$ that will be useful in formulating a graphical condition for structural target controllability. Specifically, we let
\begin{align}\label{eq:partition-Ti}
\Tc_i = \Tc_{i}^a \cup \Tc_{i}^c,
\end{align}
where:
\begin{enumerate}
\item $\Tc_{i}^a$ consists of nodes in $\Tc_i$ that either:
\begin{enumerate}
\item are not backward-connected to cycles, or
\item have at least one child in $\Tc_i$;
\end{enumerate}
\item $\Tc_{i}^c$ contains the remaining nodes in $\Tc_i$, i.e., $\Tc_{i}^c = \Tc_i \setminus \Tc_{i}^a$. All these nodes are  backward-connected to cycles.
\end{enumerate}

The following result provides a graphical criterion to determine if a set of nodes is structurally target controllable.\smallskip

\begin{theorem}{\em (Target controllability of elementary Christmas trees)}\label{thm:target_contr}
Consider an elementary Christmas tree and a set of nodes $\Tc = \Tc_0 \cup \cdots \cup \Tc_{r-1}$, where $\Tc_i \subseteq \Xc_i$. Partition each $\Tc_i$ according to \eqref{eq:partition-Ti}. The set $\Tc$ is structurally target controllable if, for all $i = 0, \dots, r-1$, the following conditions hold:
\begin{enumerate}
\item $\ell_{\text{min}}(t)\ne \ell_{\text{min}}(s)$ for all $t,s\in\Tc_{i}^a$ with $t\ne s$;
\item the nodes in $\Tc_{i}^c$ are backward-connected to different cycles.
\end{enumerate}\smallskip
\end{theorem}

\begin{proof}
The aim is to show that the system with non-zero elements in the rows of $C$ indexed by the nodes in $\Tc$ is structurally target controllable. To this end, we will show that there exists a choice of parameters of the system such that its output controllability matrix is full row-rank.

By Lemma \ref{lem:ctrb_matrix_struct}, the output controllability matrix of an elementary Christmas tree can be reduced via permutation matrices to a block diagonal matrix of the form \eqref{eq:hatRo}, where each block is associated with a periodic class. We next focus on the analysis of a single diagonal block $\hat{\Rc}_{k}^{(i)}=D^{(i)}\Rc_{k}^{(i)}$. Suppose that there exists a non-zero choice of parameters such that for every $k$ sufficiently large $\Rc_{k}^{(i)}$ has full row-rank. This guarantees that there exists a non-zero choice of parameters such that for every $k$ sufficiently large $\hat{\Rc}_{k}^{(i)}$ has full row-rank, since the diagonal entries of $D^{(i)}$ are different from zero for any choice of non-zero parameters. 
As a result, there exists $k\geq n$ and a non-zero choice of parameters such that $\hat{\Rc}_{k}^{(i)}$ has full row-rank for all $i$. For this choice of parameters 
$\hat{\Rc}_k$ has full row-rank. Since $\hat{\Rc}_k$ is related to the output controllability matrix ${\Rc}_k$ by permutation matrices, then also ${\Rc}_k$ has full row-rank.

Hence, to complete the proof, it remains to show that there exists a non-zero choice of parameters such that $\Rc_k^{(i)}$ has full row-rank. Using the characterization of $\Rc_{k}^{(i)}$ in Lemma \ref{lem:ctrb_matrix_struct}, if $t\in \Tc^a_i$ is not backward-connected to cycles, then the row of $\Rc_{k}^{(i)}$ corresponding to node $t$ is of type T1, that is, a canonical vector. Instead, if $t\in \Tc^a_i$ is backward-connected to a cycle, then, by definition of the set $\Tc^a_i$, it must have a child $s\in \Tc_i$. In this case, the rows of $\Rc_{k}^{(i)}$ corresponding to node $t$ and $s$ are of type T2 and have, respectively, the form
\begin{align*}
    r_t&=\begin{bmatrix}
        0 & \cdots & 0 & 1 & \alpha & \alpha^2 & \alpha^3 & \cdots &
    \end{bmatrix}, \\
    r_s&=\begin{bmatrix}
        0 & \cdots & 0 & 0 & 1 & \alpha & \alpha^2 & \alpha^3 & \cdots  & 
    \end{bmatrix},
\end{align*}
where $\alpha$ denotes the weight of cycle to which $t$ and $s$ are backward-connected. Note that $r_t-\alpha r_s$ is a canonical vector.
We can apply the above elementary row reduction procedure to all nodes in $\Tc^a_i$ with children, starting from nodes with shortest distance from the input $u$. In the resulting reduced matrix, each row corresponding to a node $t\in\Tc^a_i$ is a canonical vector with entry $1$ in position $\ell_{\text{min}}^{(i)}(t)\leq \lceil n/r \rceil$. Note that all these canonical vectors are different since, by condition 1, the nodes in $\Tc^a_i$ are connected to the input via shortest paths of different lengths. So, up to a permutation of rows, the reduced matrix has the form
\begin{align}\label{eq:czi_struct}
\Check{\Rc}_{k}^{(i)} = \begin{bmatrix}
            E_i & 0\\
            * & R_i
        \end{bmatrix},
\end{align}
where $E_i$ is a matrix with $|\Tc^a_i|$ rows and $s\leq \lceil n/r \rceil$ columns
which are all different canonical vectors. $R_i$ has $|\Tc^c_i|$ rows corresponding to nodes in $\Tc^c_i$. Each row of $\Check{\Rc}_{k}$ corresponding to a node $t \in \Tc^c_i$ is of type T2, and has zero entries only up to position $\ell_{\text{min}}^{(i)}(t)$. 
The inequality $\ell_{\text{min}}^{(i)}(t)\leq \lceil n/r \rceil$, $\forall t\in\Tc^c_i$, implies that for every $k$ sufficiently large $R_i$ contains a square submatrix of the form $\Delta_i V_i$
where $V_i$ is a square matrix of the Vandermonde form
\begin{align*}
        V_i=\begin{bmatrix}
             1 & \alpha_{1} & \alpha_{1}^2 & \dots &\\
             1 & \alpha_{2}& \alpha_{2}^2 & \cdots &  \\
           \vdots & \vdots  & \vdots & \vdots\\
            1 & \alpha_{|\Tc^c_i|}& \alpha_{|\Tc^c_i|}^2 &  \cdots
        \end{bmatrix},
\end{align*}
with $\alpha_i$ denoting the weights of the cycles to which the nodes in $\Tc^c_i$ are backward-connected, and $\Delta_i$ is a diagonal matrix, whose $i$-th diagonal entry is a non-negative power of $\alpha_i$. 

Since all nodes in $\Tc^c_i$ are backward-connected to different cycles, by condition 2, there exists a choice of parameters of the system such that $\alpha_i\ne \alpha_j$ for all $i\ne j$. For such choice of parameters, $V_i$ is invertible \cite[Sec.~0.9.11]{horn2012matrix}, and consequently $\Check{\Rc}_{k}^{(i)}$ has full row-rank, for every $k$ sufficiently large. To conclude, note that, since $\Check{\Rc}_{k}^{(i)}$ is related to $\Rc_{k}^{(i)}$ through left multiplication by an invertible matrix, also $\Rc_{k}^{(i)}$ has full row-rank. 
\end{proof}

The following corollary of Theorem \ref{thm:target_contr} gives a simpler graphical condition for structural target controllability. 
\begin{corollary}\label{cor:target_contr}
Consider an elementary Christmas tree and a set of nodes $\Tc = \Tc_0 \cup \cdots \cup \Tc_{r-1}$, where $\Tc_i \subseteq \Xc_i$. The set $\Tc$ is structurally target controllable if, for all $i = 0, \dots, r-1$:
\begin{enumerate}
    \item $\ell_{\text{min}}(t) \neq \ell_{\text{min}}(s)$ for all nodes $t,s\in \Tc_i$, $t\ne s$, that are not backward-connected to cycles, and
    \item all the remaining nodes in \(\Tc_i\) are backward-connected to different cycles.
\end{enumerate} 
\end{corollary}\smallskip
\begin{proof}
The result follows directly from Theorem \ref{thm:target_contr} by noting that $\Tc_i$ satisfying conditions 1 and 2 is described by a partition \eqref{eq:partition-Ti} in which $\Tc_i^a$ contains only nodes that are not backward-connected to cycles.
\end{proof}

\begin{remark}
Although Theorem \ref{thm:target_contr} (and Corollary \ref{cor:target_contr}) is stated for elementary Christmas trees, it can also be used to determine whether a set of nodes is structurally target controllable in structured systems described by a general graph $\Gc$. In fact, the conditions of Theorem \ref{thm:target_contr} can be applied to any set of nodes of $\Gc$ {\em covered} by an elementary Christmas tree. Note in particular that, since a directed tree is a particular instance of an elementary Christmas tree, these conditions apply to sets of nodes covered by directed trees. In this case, all sets $\Tc_i^c$ are empty and Theorem \ref{thm:target_contr} is equivalent to the $k$-walk condition of Proposition \ref{prop:kwalk}.
\end{remark}

In the following example, we show that Theorem \ref{thm:target_contr} allows us to find structurally target controllable sets that cannot be found using the $k$-walk and stem-cycle conditions discussed in Section \ref{sec:existing}.

\begin{example}\label{ex:thm}
We consider the structured system described by the graph $\Gc$ in Fig.~$\ref{fig:adv_ex}$. We verify that $\Tc=\{1,2,4,5,6,11,13,14\}$ is structurally target controllable by showing that it meets the conditions outlined in Theorem \ref{thm:target_contr}. Since the system features two periodic classes ($r=2$), we partition the target set $\Tc$ in two sets
\begin{align*}
    \Tc_0=\{1,5,11,13\},\ \Tc_1=\{2,4,6,14\}.
\end{align*}
Each of these sets are further partitioned according to \eqref{eq:partition-Ti}:
\begin{align*}
    \Tc_0^a=\{1,13\}, \ \Tc_0^c=\{5,11\}, \ \Tc_1^a=\{2,4\}, \ \Tc_1^c=\{6,14\}.
\end{align*}
Note that $\Tc_0^a$ contains one node not backward-connected to cycles (node $1$) and one node with a child (node $13$ with child $11$). Instead, $\Tc_1^a$ contains
two nodes with children  (node $2$ with child $4$, and node $4$ with child $6$). The sets $\Tc_0^a$ and $\Tc_1^a$ satisfy condition 1 in Theorem \ref{thm:target_contr}; in fact, all nodes in these sets are reached by shortest paths of different lengths. 
Further, the nodes in $\Tc_0^c$ are backward-connected to different cycles and the same holds for $\Tc_1^c$. Therefore, condition 2 of Theorem \ref{thm:target_contr} is also satisfied, and we conclude that $\Tc$ is structurally target controllable.

\begin{figure}
    \centering
    \begin{tikzpicture}[shorten >=1pt, auto, ultra thick,
   node_style/.style={draw, circle,thick, fill=white, minimum size=0.5cm,font=\footnotesize,inner sep=0cm},every edge/.append style = {thick}]
            
        \node[node_style] (u) at (0,0) {$u$};
        \node[node_style,] (1) at (1,1.2) {$1$};
        \node[node_style,fill=black, text=white] (2) at (1+1.2,1.2) {$2$};
        \node[node_style] (3) at (1+2*1.2,1.2) {$3$};
        \node[node_style,fill=black, text=white] (4) at (1+3*1.2,1.2) {$4$};
        \node[node_style] (5) at (1+4*1.2,1.2) {$5$};
        \node[node_style,fill=black, text=white] (6) at (1+5*1.2,1.2) {$6$};

        \node[node_style] (7) at (1,-1.2) {$7$};
        \node[node_style,fill=black, text=white] (8) at (1+1.2,-1.2) {$8$};
        \node[node_style] (9) at (1+2*1.2,-1.2) {$9$};
        \node[node_style,fill=black, text=white] (10) at (1+3*1.2,-1.2) {$10$};
        \node[node_style] (11) at (1+4*1.2,-1.2) {$11$};
        \node[node_style,fill=black, text=white] (12) at (1+5*1.2,-1.2) {$12$};
        
        \node[node_style] (13) at (1+2*1.2,0) {$13$};
        \node[node_style,fill=black, text=white] (14) at (1+3*1.2,0) {$14$};

        \draw[-latex,semithick,line join=round,
             decorate, decoration={
                zigzag,
                segment length=4,
                amplitude=.9,
                post=lineto,
                post length=7.5pt
            }] (u) -- (1);
            \draw[-latex,semithick,line join=round,
             decorate, decoration={
                zigzag,
                segment length=4,
                amplitude=.9,
                post=lineto,
                post length=7.5pt
            }] (u) -- (7);
            
        \draw[-latex] (1) edge node{} (2);
        \draw[-latex] (2) edge[bend left] node{} (3);
        \draw[-latex] (3) edge node{} (4);
        \draw[-latex] (5) edge node{} (6);
        \draw[-latex] (4) edge node{} (5);

        \draw[-latex] (7) edge node{} (8);
        \draw[-latex] (8) edge[bend left] node{} (9);
        \draw[-latex] (9) edge node{} (10);
        \draw[-latex] (10) edge node{} (11);
        \draw[-latex] (11) edge node{} (12);

        \draw[-latex] (8) edge node{} (13);
        \draw[-latex] (13) edge node{} (14);
        
        \draw[-latex] (9) edge[bend left] node{} (8);
        \draw[-latex] (3) edge[bend left] node{} (2);
        
    \end{tikzpicture}
    \vspace{-0.25cm}
    \caption{Graph $\Gc$ of an elementary Christmas tree network used for comparing the conditions in Theorem \ref{thm:target_contr} with the existing structural target controllability conditions. White nodes belong to the set $\Xc_1 \subset \Xc$, black nodes are in the set $\Xc_0 \subset \Xc$. }
    \label{fig:adv_ex}
\end{figure}
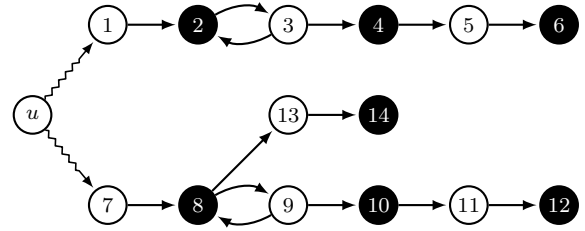

Notice that $\Tc$ does not satisfy any of the existing conditions in Sec.~\ref{sec:existing}. Indeed, nodes $4$ and $14$ are reached by shortest paths of equal length, and the same applies to nodes $5$ and $11$. This implies that structural target controllability cannot be verified with the $k$-walk condition of Proposition \ref{prop:kwalk}. Moreover, the same nodes $4$, $5$, $11$, $14$ cannot be covered by a disjoint union of stems and cycles; consequently, the stem-cycle condition of Proposition \ref{prop:stemcycle} does not apply either.~\exampleend
\end{example}

To conclude this section, we point out that Theorem \ref{thm:target_contr} provides a sufficient but not necessary condition for a set to be structurally target controllable. In fact, there may exist sets of nodes that are structurally target controllable but do not satisfy the conditions of Theorem \ref{thm:target_contr}. For instance, for the system of Example \ref{ex:thm}, the set $\Tc=\{2,6,8,12\}$ can be shown to be structurally target controllable. In this case, with reference to the partition in equation \eqref{eq:partition-Ti}, all nodes in $\Tc$ belong to $\Tc_1^c$, and there are couples $\{2,6\}$, $\{8,12\}$ whose nodes are backward-connected to the same cycles. Thus, condition 2 of Theorem  \ref{eq:partition-Ti} is not satisfied for such $\Tc$.

\section{Conclusions}\label{sec:conclusions}
In this work, we focused on structural target controllability of directed networks. 
After providing a brief overview of existing graph-theoretic conditions, we introduced Christmas trees, a new class of network systems. 
Christmas trees represent one of the simplest classes of input-accessible networks that contains non-structurally controllable systems. Consequently, determining whether a subset of their nodes is target controllable is non-trivial.

In the paper, we took the first steps in the controllability analysis of these networks starting from their simplest instance, that is, elementary Christmas trees. 
Specifically, we derived graph-theoretic conditions for structural target controllability in elementary Christmas trees. Our conditions extend existing ones, enabling the identification of target controllable node sets that previous criteria fail to detect. Moreover, they can be applied to general network systems using spanning subgraphs of Christmas tree class.

We believe that the study of this class can provide valuable insights for developing new and more comprehensive conditions for structural target controllability, just as cacti have been essential in understanding full-state structural controllability.
As pointed out in an example, our conditions are only sufficient for structural target controllability.
Therefore, future research directions include finding a complete characterization of the target controllable sets in elementary Christmas trees, and extending such characterization to general Christmas trees. We believe that tackling these problems will ultimately lead to a more complete characterization of structural target controllability in general network systems.

\balance
\bibliographystyle{ieeetr}
\bibliography{bib-latex}

\end{document}